\documentclass[preprint,12pt]{article}

\usepackage{graphicx}

\usepackage{amssymb}
\usepackage{amsthm}
\usepackage{amsmath}

\begin{document}




\title{An extension of the Georgiou-Smith example: Boundedness and attractivity  in the presence of unmodelled dynamics via nonlinear PI control}


\author{Haris E. Psillakis\\
National Technical University of Athens (NTUA)\\
H. Polytechniou 9,
15780 Zografou, Athens, Greece\\
\texttt{hpsilakis@central.ntua.gr}}
\date{}
\maketitle
\begin{abstract}
In this paper, a nonlinear extension of the Georgiou-Smith system is considered and robustness results are proved for a class of nonlinear PI controllers with respect to fast parasitic first-order dynamics. More specifically, for a perturbed nonlinear system with sector bounded nonlinearity and unknown control direction, sufficient conditions for global boundedness and attractivity have been derived. It is shown that  the closed loop system is globally bounded and attractive if (i) the unmodelled dynamics are sufficiently fast and (ii) the PI  control gain has the Nussbaum function property. For the case of nominally unstable systems, the Nussbaum  property of the control gain appears to be crucial. A simulation study confirms the theoretical results.
\end{abstract}

\section{Introduction}
\label{intro}
\newtheorem{theorem}{Theorem}
\newtheorem{lemma}{Lemma}
\newtheorem{remark}{Remark}
\newtheorem{definition}{Definition}
\newtheorem{assumption}{Assumption}
\newtheorem{corollary}{Corollary}
The unknown control direction problem has attracted significant research interest over the last three decades. Nussbaum gains \cite{Nussbaum_paper}, \cite{Ye_Jiang98} originally introduced in \cite{Nussbaum_paper} have become the main theoretical tool for controller design for systems with unknown control directions.
Nussbaum functions (NFs) are continuous functions $N(\cdot)$ with the property
\begin{align}
    \limsup_{\zeta\rightarrow +\infty}&\frac{1}{\zeta}\int_{0}^{\zeta}{N(s)ds}=+\infty\label{nussbaum propertyp}\\
    \liminf_{\zeta\rightarrow +\infty}&\frac{1}{\zeta}\int_{0}^{\zeta}{N(s)ds}=-\infty\label{nussbaum propertym}.
\end{align}
Examples of NFs are $\zeta^2\cos(\zeta)$, $\exp(\zeta^2)\sin(\zeta)$ among many others.

For the simple integrator case $\dot{y}=bu$ with $b$ a nonzero constant of \emph{unknown sign}, standard analysis \cite{Nussbaum_paper} shows that the Nussbaum control law
\begin{align}\label{nussbaum controller}
 u&=\zeta^2\cos(\zeta)y\nonumber\\
 \dot{\zeta}&=y^2
\end{align}
ensures convergence of the output $y$ to the origin and boundedness of the Nussbaum parameter $\zeta$. However, Georgiou and Smith demonstrated in \cite{GS} that the proposed controller is nonrobust to fast parasitic unmodelled dynamics. Particularly, they considered the  system
\begin{align}\label{G-S}
    \dot{x}&=bu\nonumber\\
    \dot{y}&=M(x-y)
\end{align}
and showed divergence for $M>1$ when the controller \eqref{nussbaum controller} is used.

An alternative nonlinear PI methodology was proposed by Ortega, Astolfi and Barabanov  in \cite{Ortega_paper} to address the unknown control direction problem. For the simple integrator case, their controller takes the form
\begin{equation}\label{nonlinear PI}
    u=z\cos(z)y
\end{equation}
with $z$ the PI square error defined by $z=(1/2)y^2+\lambda\int_0^t{y^2(s)ds}$. The main difference between the two controllers \eqref{nussbaum controller}, \eqref{nonlinear PI}  is the existence of the proportional term in the control gain of \eqref{nonlinear PI} (see also p. 166 of \cite{AKO_book}). It was hinted  in \cite{Ortega_paper},\cite{AKO_book}(no complete proof was given) that such a controller is robust to fast parasitic first-order perturbations and therefore can stabilize the Georgiou-Smith example system if $\lambda<M$. Their argument, however, was based on the fact that the related transfer function is  positive real and cannot be carried over to the case of an unstable unforced linear system or even a nonlinear system. In fact, the introduction of a simple destabilizing pole in the system 
\begin{align}\label{G-SL}
    \dot{x}&=\alpha x+bu\nonumber\\
    \dot{y}&=M(x-y)
\end{align}
($\alpha>0$) may result in instability of the closed-loop system with the controller \eqref{nonlinear PI} even if $\alpha+\lambda<M$ (see Section \ref{simulation}).

It remains therefore an open problem to design a nonlinear PI controller   robust to fast parasitic dynamics when the plant to be controlled is originally unstable and nonlinear. To this end, we consider an extension of the Georgiou-Smith  system. Particularly, we examine the overall dynamic behavior  of  the  nonlinear system  with first-order unmodelled dynamics given by
\begin{align}\label{G-SE}
    \dot{x}&=f(x)+bu\nonumber\\
    \dot{y}&=M(x-y)
\end{align}
when a nonlinear PI control law $u$ designed for the unperturbed system
\begin{equation}\label{unperturbed}
    \dot{y}=f(y)+bu
\end{equation}
is applied.
\subsection{Nonlinear PI for the unperturbed system}
For system \eqref{unperturbed}, we assume that $f(\cdot)$
is a sector-bounded nonlinearity, i.e. $f(0)=0$, $f(y)=y\alpha(y)$ and there exist some constants $\alpha_1, \alpha_2\in\mathbb{R}$ such that $\alpha_1\leq \alpha(y)\leq \alpha_2$ $\forall y\in\mathbb{R}$.
A controllability assumption is also imposed, that is $b\neq 0$.
Let now a nonlinear PI controller of the form
\begin{align}
    u&=\kappa(z)y\label{PI_o}\\
    z&=(1/2)y^2+\lambda\int_0^t{y^2(s)ds}\label{z}
\end{align}
with PI gain $\kappa(z)=\alpha_0(z)\cos(z)$ where $\alpha_0(\cdot)$ is a class $\mathcal{K}_{\infty}$ function and $\lambda>0$.
For the $z$ derivative we have for the unperturbed system \eqref{unperturbed} that
\begin{equation}\label{zdot}
    \dot{z}\leq \big[\max\{|\alpha_1|,|\alpha_2|\}+\lambda+b \kappa(z)\big]y^2.
\end{equation}
Note that  whenever $z(t)=z_k$ with $z_k:=(\pi/2)[4k+1+\textrm{sgn}(b)]$ we have
\begin{equation}\label{zdot}
    \dot{z}(t)\leq \big[\max\{|\alpha_1|,|\alpha_2|\}+\lambda-|b|\alpha_0(2k\pi)\big]y^2(t).
\end{equation}
Thus, $\dot{z}(t)\leq 0$ whenever $z(t)=z_k$ for every $k\geq k_0$
\begin{equation*}
  k_0:=\bigg\lceil\frac{1}{2\pi}\alpha_0^{-1}\bigg[\frac{1}{|b|}(\max\{|\alpha_1|,|\alpha_2|\}+\lambda)\bigg]\bigg\rceil
\end{equation*}
($\lceil x\rceil$ denotes the largest integer not exceeding $x$) which in turn implies that $z$ is bounded by $z(t)\leq z_{k'}$ where $k':=\max\{k_0,\lceil y^2(0)/4\pi\rceil\}$. The fact that $z\in\mathcal{L}_{\infty}$ implies $y\in\mathcal{L}_{\infty}\cap\mathcal{L}_{2}$ and $u\in\mathcal{L}_{\infty}$ from \eqref{z} and \eqref{PI_o} respectively. Also, from \eqref{unperturbed} we have $\dot{y}\in{\mathcal{L}_{\infty}}$. Barbalat's lemma can now be invoked to prove that $\lim_{t\rightarrow\infty}y(t)=0$. This is a standard analysis in the spirit of  \cite{AKO_book}.

Assume now the existence of parasitic first order unmodelled dynamics in the form of \eqref{G-SE}.  Sufficient conditions are given in the next section for \emph{global boundedness and attractivity} for the closed-loop system comprised from \eqref{G-SE} and the nonlinear PI controller  \eqref{PI_o} and \eqref{z}. A key property is that the nonlinear PI gain function $\kappa(z)$ should be a function of Nussbaum type. 

\section{Extended Georgiou-Smith system with sector nonlinearity}
\label{}
In this section we consider system \eqref{G-SE} with a  sector-bounded nonlinearity
\begin{align}
    f(x)&=\alpha(x)x\label{sector}\\
\alpha_1\leq \alpha&(x)\leq \alpha_2 \qquad \forall x\in\mathbb{R}\label{sector bounds}
\end{align}
for some constants $\alpha_1,\alpha_2\in\mathbb{R}$. Note that $\alpha_1,\alpha_2$ can also take positive values rendering the unforced system unstable. 
To simplify notation let us define the constant $\epsilon:=1/M$.
We have established the following theorem.
\begin{theorem}\label{main_theorem}
Let the closed-loop system described by \eqref{G-SE}, \eqref{PI_o}, \eqref{z} with sector-bounded nonlinearity given by \eqref{sector}, \eqref{sector bounds}. If
\begin{description}
  \item[(i)] $\epsilon\lambda<1$, $\epsilon(\lambda+\alpha_2)<1$
  \item[(ii)] $\alpha_2-\alpha_1\leq \frac{2\lambda}{\sqrt{1-\epsilon\lambda}}\Big[\sqrt{1-\epsilon(\lambda+\alpha_1)}+\sqrt{1-\epsilon(\lambda+\alpha_2)}\Big]$
  \item[(iii)] $\kappa(z)$ has the  Nussbaum property \eqref{nussbaum propertyp},\eqref{nussbaum propertym}
\end{description}
then, all closed-loop signals are bounded and $\lim_{t\rightarrow\infty}y(t)=\lim_{t\rightarrow\infty}x(t)=0.$
\end{theorem}
\begin{proof}
From the definition of the PI error $z$ in \eqref{z} and \eqref{G-SE} we have that
\begin{equation}\label{z_dynamics}
    \dot{z}=Mxy-M(1-\epsilon\lambda)y^2.
\end{equation}
Let now the function
\begin{align}\label{S}
    S:=\frac{\lambda}{2}x^2+\frac{1}{2}M(1-\epsilon\lambda)(x-y)^2+\epsilon c z- b\int_0^z{\kappa(s)ds}
\end{align}
with $c\in\mathbb{R}$ to be defined. Replacing from \eqref{G-SE}, \eqref{PI_o}, \eqref{z}, \eqref{z_dynamics} and canceling terms we have for its time derivative that
\begin{align}\label{dotS}
    \dot{S}=\lambda\alpha(x)x^2-M^2&(1-\epsilon\lambda)(x-y)^2\nonumber\\
    &+M(1-\epsilon\lambda)\alpha(x)x(x-y)+cxy-c(1-\epsilon\lambda)y^2.
\end{align}
Eq. \eqref{dotS} can be written in matrix notation as
\begin{align}\label{dotS1}
    \dot{S}&=-M^2\left[
                                   \begin{array}{cc}
                                     x & y \\
                                   \end{array}
                                 \right]
    \left[
    \begin{array}{cc}
    1-\epsilon(\lambda+\alpha(x)) & -\frac{1}{2}[c\epsilon^2+(1-\epsilon\lambda)(2-\epsilon\alpha(x))]\\
    * & (1-\epsilon\lambda)(c\epsilon^2+1)
    \end{array}\right]\left[
                        \begin{array}{c}
                          x \\
                          y \\
                        \end{array}
                      \right]
                      \nonumber\\
         & := -M^2\left[
                                   \begin{array}{cc}
                                     x & y \\
                                   \end{array}
                                 \right]\Lambda(x)  \left[
                        \begin{array}{c}
                          x \\
                          y \\
                        \end{array}
                      \right]
\end{align}
where $*$  denotes a symmetric w.r.t. the main diagonal element of $\Lambda(x)$. We claim that there is some constant $c\in\mathbb{R}$ such that $\Lambda(x)$ is positive definite for all $x\in\mathbb{R}$. Equivalently, we can prove that, for some $c\in\mathbb{R}$, the two principal minors of $\Lambda(x)$ given by
\begin{align*}
   \Delta_1(x)&:=1-\epsilon(\lambda+\alpha(x))\\
   \Delta_2(x)&:=[1-\epsilon(\lambda+\alpha(x))](1-\epsilon\lambda)(c\epsilon^2+1)-\frac{1}{4} [c\epsilon^2+(1-\epsilon\lambda)(2-\epsilon\alpha(x))]^2
\end{align*}
are positive $\forall x\in\mathbb{R}$. From assumption (i)  of Theorem \ref{main_theorem}, it is obvious that $\Delta_1(x)>0$ $\forall x\in\mathbb{R}$. For $\Delta_2(x)$ we have that
\begin{equation}\label{Delta_2}
    \Delta_2(x)=-(1/4)\big[A_c c^2+B_c(x) c+\Gamma_c(x)\big]
\end{equation}
with $A_c=\epsilon^4>0$, $B_c(x)=2\epsilon^3(1-\epsilon\lambda)[\alpha(x)+2\lambda]$, $\Gamma_c(x)=\epsilon^2(1-\epsilon\lambda)\alpha(x)[4\lambda+(1-\epsilon\lambda)\alpha(x)]$. $\Delta_2(x)$ is therefore a quadratic polynomial with respect to $c$ that is positive definite  if
\begin{description}
  \item[(a)] $\Delta_c(x):=B_c(x)^2-4A_c\Gamma_c(x)>0$
  \item[(b)] there exists some constant $c\in \big(c_1(\alpha(x)),c_2(\alpha(x))\big)$ for all $x\in\mathbb{R}$ where $c_1(\alpha(x))$, $c_2(\alpha(x))$ are the two roots of $\Delta_2(x)$ given by
\begin{align*}
    c_1(\alpha(x))&:=-M\Big[(1-\epsilon\lambda)(2\lambda+\alpha(x))+2\lambda\sqrt{(1-\epsilon\lambda)\big[1-\epsilon(\lambda+\alpha(x))\big]}\Big]\\
    c_2(\alpha(x))&:=-M\Big[(1-\epsilon\lambda)(2\lambda+\alpha(x))-2\lambda\sqrt{(1-\epsilon\lambda)\big[1-\epsilon(\lambda+\alpha(x))\big]}\Big].
\end{align*}
\end{description}
If we carry out the calculations we have that
\begin{equation*}
    \Delta_c(x)=16\epsilon^6\lambda^2(1-\epsilon\lambda)[1-\epsilon(\lambda+\alpha(x))]
\end{equation*}
and therefore the positivity  condition for $\Delta_c(x)$ is satisfied if $\epsilon(\lambda+\alpha_2)<1$ and $\epsilon\lambda<1$. For condition (b) to be true, as $\alpha(x)$ varies in $[\alpha_1,\alpha_2]$, there must be some $c\in\mathbb{R}$ such that $c\in[c_1(\alpha),c_2(\alpha)]$ for all $\alpha\in[\alpha_1,\alpha_2]$. This holds true if $\max_{\alpha\in[\alpha_1,\alpha_2]}c_1(\alpha)<\min_{\alpha\in[\alpha_1,\alpha_2]}c_2(\alpha)$. Function $c_2(\cdot)$ is obviously decreasing with respect to $\alpha(x)$ with minimum value $c_2(\alpha_2)> c_2((1/\epsilon)(1-\epsilon\lambda))=-(1/\epsilon^2)(1-\epsilon^2\lambda^2)$. Function $c_1(\cdot)$ on the other hand is decreasing up to some point $\alpha_0=(1/\epsilon)[1-\epsilon\lambda/(1-\epsilon\lambda)]$ and then increasing up to $(1/\epsilon)(1-\epsilon\lambda)$ with value $c_1((1/\epsilon)(1-\epsilon\lambda))=-(1/\epsilon^2)(1-\epsilon^2\lambda^2)<c_2(\alpha_2)$. Thus, the second condition holds true if $c_1(\alpha_1)<c_2(\alpha_2)$ which is exactly assumption (ii) of the theorem. Thus, selecting  $c:=\epsilon_0 c_1(\alpha_1)+(1-\epsilon_0)c_2(\alpha_2)$ for any $\epsilon_0\in(0,1)$  we have $\dot{S}\leq 0$.
Integrating now $\dot{S}\leq 0$ we have that $S(t)\leq S(0)$ or equivalently
\begin{align}\label{Sbound}
    \lambda x^2(t)+M(1-\epsilon\lambda)(x(t)-y(t))^2 <2S(0)-2\epsilon cz(t)+ 2b\int_0^{z(t)}{\kappa(s)ds}.
\end{align}
The above inequality and the Nussbaum property of $\kappa(\cdot)$ ensure the boundedness of $z$. To prove this, let us assume the contrary. From the Nussbaum property (iii) of $\kappa(z)$ there exists a strictly increasing sequence $\{z_k\}_{k=1}^{\infty}$ such that $\lim_{k\rightarrow\infty}z_k=+\infty$ and
\begin{equation}\label{nussbaum_zk}
    \lim_{k\rightarrow\infty}\frac{1}{z_k}\int_0^{z_k}{b\kappa(w)dw}=-\infty.
\end{equation}
Due to continuity of $z$, the PI square error $z$ will eventually pass from an infinite number of elements of $\{z_k\}_{k=1}^{\infty}$. From \eqref{Sbound}, we have for the times $t_k$ at which $z(t_k)=z_k$
\begin{align}\label{Sbound_tk}
   \frac{\lambda}{2} x^2(t_k)+\frac{1}{2}M(1-\epsilon\lambda)(x(t_k)-y(t_k))^2 <S(0)-\epsilon cz_k+ b\int_{0}^{z_k}{\kappa(w)dw}.
\end{align}
If we divide all terms in \eqref{Sbound_tk} with $z_k$ and take into account  the limiting property \eqref{nussbaum_zk} the left hand side (l.h.s.) of \eqref{Sbound_tk} should take negative values for all $k\geq k_0$ for some $k_0\in\mathbb{N}$. This yields the desired contradiction since the l.h.s. of \eqref{Sbound_tk} is a sum of squares which is always nonnegative. Thus, $z\in\mathcal{L}_{\infty}$ and therefore $y\in\mathcal{L}_{\infty}\cap\mathcal{L}_{2}$, $u\in\mathcal{L}_{\infty}$ and from \eqref{Sbound} $x\in\mathcal{L}_{\infty}$. Then, the system equations \eqref{G-SE} yield $\dot{x},\dot{y}\in\mathcal{L}_{\infty}$. Invoking now Barbalat lemma we obtain the desired property $\lim_{t\rightarrow\infty}x(t)=\lim_{t\rightarrow\infty}y(t)=0$.
\end{proof}
\begin{remark}
In the case of a linear system $f(x)=\alpha x$ condition (ii) of Theorem \ref{main_theorem} is no longer needed and  (i)  reduces to $\epsilon\lambda<1$ for $\alpha\leq 0$ (nominally stable system) and $\epsilon(\alpha+\lambda)<1$ for $\alpha>0$ (nominally unstable). Note that in the latter case the necessary condition for stabilization by simple output feedback (with known sign of $b$) is $\epsilon\alpha<1$.
 \end{remark}
 \begin{remark}
 If $c_2(\alpha_2)\geq 0$ then the constant $c$ in the definition of $S$ can be nonnegative. This means that in \eqref{Sbound}  $-c\epsilon z(t)\leq 0$ and the Nussbaum condition (iii) for $\kappa(z)$ in Theorem \ref{main_theorem} can be relaxed to
\begin{align}
    \limsup_{z\rightarrow\infty}\int_{0}^{z}{\kappa(s)ds}=+\infty\\
    \liminf_{z\rightarrow\infty}\int_{0}^{z}{\kappa(s)ds}=-\infty.
\end{align}
After calculations one can show that condition $c_2(\alpha_2)\geq 0$ holds true iff $\alpha_2\leq 0$.
Thus, if the unforced linear system is stable ($\alpha\leq 0$) and $\epsilon\lambda<1$ then, the controller \eqref{nonlinear PI} results in bounded and attractive closed-loop behavior. This also provides a strict proof for the integrator example of \cite{Ortega_paper}, \cite{AKO_book}.
\end{remark}
\begin{remark}\label{remark_sector}
Note that the r.h.s. of condition (ii) tends to $4\lambda$ in the limit $\epsilon\rightarrow 0$. Thus,  for some sector bounded nonlinearity \eqref{sector bounds}, if we select $\lambda>(\alpha_2-\alpha_1)/4$ then there exists some $\epsilon_0>0$ such that for all $\epsilon<\epsilon_0$ the closed-loop system \eqref{G-SE}, \eqref{PI_o}, \eqref{z} is globally bounded and attractive.
\end{remark}
\section{Simulation results}
 \label{simulation}
A simulation study was performed for the perturbed integrator  (P-INT) and the perturbed linear system  (P-LS) described by  \eqref{G-S}, \eqref{G-SL} respectively with parameters $\alpha=1$, $b=1/2$, $\lambda=2.5$, $\epsilon=1/4$ and initial conditions $x(0)=y(0)=4$. For the specific parameters, condition (i) of Theorem \ref{main_theorem} holds true. We tested the case of a Nussbaum gain based (NG) controller \eqref{nussbaum controller} and a nonlinear PI controller \eqref{nonlinear PI} with gains $\kappa(z)=z\cos(z)$  (not a Nussbaum function) denoted as nPI and $\kappa(z)=z^2\cos(z)$ (Nussbaum function) denoted as nPI-N.
\begin{figure}[!t]
\centering
\includegraphics[width=4.1in]{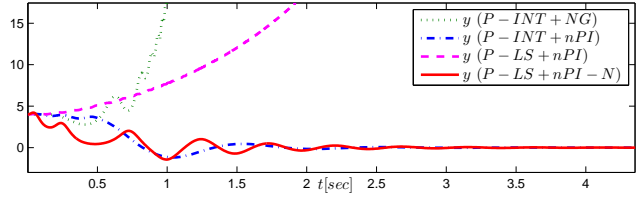}
\caption{Linear system: Time responses of  $y(t)$ for the cases of a perturbed integrator (P-INT) and a perturbed linear system (P-LS) for the three controllers NG, nPI, nPI-N.}
\label{comparisons}
\end{figure}
The output response $y$ shown in Fig. \ref{comparisons} verifies our theoretical analysis. Particularly, for the P-INT system with the NG controller, $y$ is divergent as shown in \cite{GS}. If the nPI controller is used then $y$ remains bounded and converges to zero \cite{Ortega_paper}, \cite{AKO_book}. However, the nPI control fails to regulate the P-LS system. Convergent solutions are obtained for the P-LS system only when the nPI-N is employed.

Let now the perturbed nonlinear system \eqref{G-SE} with $f(x)=3[1+\sin^2(x)]x$ where $\alpha_1=3$, $\alpha_2=6$ and $b=1$. Selecting $\lambda=2.5$, we have that both conditions (i) and (ii) are satisfied for every $\epsilon<1/(\alpha_2+\lambda)=2/17$ (see Remark \ref{remark_sector}).  For the  control law \eqref{PI_o}, \eqref{z}, $\kappa(z)=z^2\sin(z)$ simulation results are shown in Fig. \ref{sector_example} with $\epsilon=0.1$  and initial conditions $x(0)=y(0)=4$.
\begin{figure}[!t]
\centering
\includegraphics[width=4.1in]{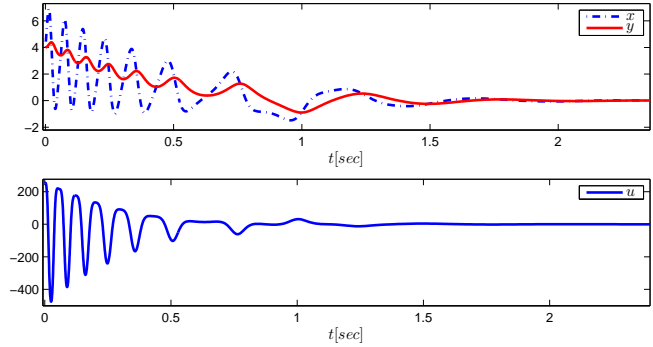}
\caption{Nonlinear system: Time responses of system states  $x,y$ and control input $u$.}
\label{sector_example}
\end{figure}
As expected, all $x,y,u$ are bounded and converge to the origin as time passes.
 \section{Conclusions}
Sufficient conditions are derived for global boundedness and attractivity of a perturbed nonlinear system with sector-bounded nonlinearity under a nonlinear PI control action.
The results further demonstrate the superiority of the nonlinear PI controls compared to simple Nussbaum gain based schemes with respect to robustness to unmodelled dynamics.



\end{document}